\documentclass[preprint]{elsarticle}

\usepackage{array,xspace,multirow,hhline,tikz,colortbl,tabularx,booktabs,fixltx2e,amsmath,amssymb,amsfonts,amsthm}
\usepackage{algorithm}
\usepackage{algorithmic}
\usepackage{verbatim,ifthen}
\usepackage{enumitem}
\usepackage{pifont}
\usepackage{ifthen}
\usepackage{calrsfs,mathrsfs}
\usepackage{bbding,pifont}
\usepackage{pgflibraryshapes}

\usetikzlibrary{trees}
\usetikzlibrary{positioning,chains,fit,shapes,calc}
\usepackage{tkz-graph}

\usepackage{eucal}

\usepackage{tikz}
\usetikzlibrary{arrows}
\usetikzlibrary{decorations.pathreplacing}

	\usepackage{varioref}

\definecolor{light-gray}{gray}{0.9}

	\newcommand{\figref}[1]{Figure~\ref{#1}}

	\newcommand{\lemref}[1]{Lemma~\ref{#1}}
	\newcommand{\thmref}[1]{Theorem~\ref{#1}}

	\newcommand{\ie}{i.e.,\xspace}
	\newcommand{\eg}{e.g.,\xspace}

	\newtheorem{lemma}{Lemma}%
	\newtheorem{theorem}{Theorem}%
	\newtheorem{corollary}{Corollary}%
	\newcommand\eat[1]{}

	\usepackage{enumitem}
	\setenumerate[1]{label=\rm(\it{\roman{*}}\rm),ref=({\it\roman{*}}),leftmargin=*}
	\newlength{\wordlength}

	\newcommand{\eqclass}[2][]{\ifthenelse{\equal{#1}{}}{[#2]}{[#2]_{\sim_{#1}}}}

	\newcommand{\Pref}[1][]{
		\ifthenelse{\equal{#1}{}}{\mathrel R}{\mathop{R_{#1}}}
	}                                          
	\newcommand{\sPref}[1][]{                  
		\ifthenelse{\equal{#1}{}}{\mathrel P}{\mathop{P_{#1}}}
	}                                          
	\newcommand{\Indiff}[1][]{                 
		\ifthenelse{\equal{#1}{}}{\mathrel I}{\mathop{I_{#1}}}
	}
	\newcommand{\prefset}[1][]{\ifthenelse{\equal{#1}{}}{\mathcal{R}}{\mathcal{R}_{#1}}}

\usepackage{enumitem}
\setenumerate[1]{label=\rm(\it{\roman{*}}\rm),ref=({\it\roman{*}}),leftmargin=*}

\newcommand{\nbh}[1][]{
	\ifthenelse{\equal{#1}{}}{\nu}{\nu(#1)}
}

\newcommand{\cstr}[1][]{
	\ifthenelse{\equal{#1}{}}{\mathscr S}{\cstr(#1)}
}

\newcommand{\choice}[1][]{
	\ifthenelse{\equal{#1}{}}{\mathit{C}}{\choice(#1)}
}

\usepackage{boxedminipage}
\usepackage{xspace}

\tikzset{
  treenode/.style = {align=center, inner sep=0pt, text centered,
    font=\sffamily},
  arn_n/.style = {treenode, circle, white, font=\sffamily\bfseries, draw=black,
    fill=black, text width=1.5em},
  arn_r/.style = {treenode, circle, red, draw=red, 
    text width=1.5em, very thick},
  arn_x/.style = {treenode, rectangle, draw=black,
    minimum width=0.5em, minimum height=0.5em}
}

\newcommand{\studentproposing}{student-proposing\xspace} 
\newcommand{\collegeproposing}{college-proposing\xspace} 

\begin{document}

\title{On the Susceptibility of the \\Deferred Acceptance Algorithm}
	\author{Haris Aziz} \ead{haris.aziz@nicta.com.au}
	\address{NICTA and UNSW, Australia}
	
	\author{Hans Georg Seedig} \ead{seedig@in.tum.de}
	\address{TU M\"unchen, Germany}
	
	\author{Jana Karina von Wedel} \ead{wedel@ma.tum.de}
	\address{TU M\"unchen, Germany}



\begin{abstract}
	The Deferred Acceptance Algorithm (DAA) is the most widely accepted and used algorithm to match students, workers,  or residents to colleges, firms or hospitals respectively. In this paper, we consider for the first time, the complexity of manipulating DAA by agents such as colleges that have capacity more than one. For such agents, truncation is not an exhaustive strategy. We present efficient algorithms to compute a manipulation for the colleges when the colleges are proposing or being proposed to. We then conduct detailed experiments on the frequency of manipulable instances in order to get better insight into strategic aspects of two-sided matching markets. Our results bear somewhat negative news: assuming that agents have information other agents' preference, they not only often have an incentive to misreport but there exist efficient algorithms to find such a misreport.
\end{abstract}


\maketitle

\section{Introduction}
\label{section:Introduction}

Deciding which students get admitted to which college, which workers get jobs at which firm, matching medical students to hospitals for their internship or residency -- all these are examples of \emph{two-sided matching markets}. In these markets, agents from two disjoint sets are paired with each other using a matching mechanism that takes into account their reported \emph{preferences} and \emph{capacities}. 
The study of such matching market mechanisms is an important field in microeconomics. Alvin Roth and Lloyd Shapley were awarded the Nobel memorial prize in economic science ``for the theory of stable allocations and the practice of market design''. Within market design, the \emph{deferred acceptance algorithm} (DAA)---sometimes referred to as the \emph{Gale-Shapley algorithm} \citep{GaSh62a}---is one of the most widely used centralized matching market mechanism.
The main focus of this paper is on DAA applied on the college admission market in which students are matched to colleges.

One of the main reasons for the popularity of DAA for real-world matching markets, is that it always yields a \emph{stable matching}. A matching is stable if no agent is assigned a partner that is unacceptable to him or her and no two agents, that are not matched with each other, prefer each other to (one of) their respective matching partners. 

However, a well-known drawback of DAA is that it is not strategyproof, \ie there exist instances for which an agent has an incentive to misreport his preferences. This does not stem from a design flaw of DAA but from the fact that no stable matching algorithm is strategyproof~\citep[see ][p. 622f]{Roth82a}.
Two redeeming factors regarding this manipulability could be \emph{(i)} that finding a beneficial manipulation strategy for a given agent, if one exists, could be \emph{computationally intractable} or  \emph{(ii)} that only an \emph{insignificantly small fraction} of instances of matching markets can be beneficially manipulated. 
The following two research questions naturally arise: \emph{Is the problem of finding a beneficial manipulation for a given agent computationally feasible? Is a significant fraction of instances of matching markets beneficially manipulable?}


For manipulations by agents that can only be matched to a single other agent, these questions have already been partially answered in the work by \citet{ImMa05a}, \citet{KoPa09a} and \citet{TST01a}.
The goal of this paper is to examine strategic issues for  agents that may have capacity more than one, \eg colleges in college admission markets.
These strategic issues can have profound effects on the assessment of stable matching procedures since the strong normative properties of DAA such as stability only hold as long as agents report their truthful preferences. If agents are frequently in a position to beneficially misreport their preferences and if finding such a misreport is computationally feasible, then the normative properties of DAA outcomes may be compromised.

We show that the problem of finding a beneficial manipulation for a given college is computationally feasible for both variants assuming they have knowledge of the preferences of all students and colleges.  
In particular, we prove that under \collegeproposing DAA, for a given college, a beneficial manipulation or the fact that no such manipulation exists can be determined in polynomial time. The same problem under \studentproposing DAA is shown to be fixed-parameter tractable w.r.t. the capacity of the college. 
Based on experimental results obtained by simulating \collegeproposing and \studentproposing DAA%
, we further conclude that a significant fraction of instances of matching markets can be beneficially manipulated by at least one college. 

%
%
%

\section{Related Work}
\label{section:relatedWork}
\citet{Roth82a} showed that there exists no stable and strategyproof matching mechanism for two-sided matching markets. 
This impossibility theorem inspired several avenues of research relevant to this paper. 
Some researchers tried to determine the relevance of this result for real-world applications of DAA by investigating which \emph{fraction of instances} of matching markets is in fact beneficially manipulable~\citep{ImMa05a,KoPa09a,RoPe99a,TST01a}. 
We complement this existing work by considering manipulation by colleges via misrepresenting preferences in markets of the college admission type in which agents are required to submit \emph{complete preference lists}. 
Regarding simulations, we use Mallows mixture models instead of uniform distributions to generate preferences whose structure is arguably more realistic.
An orthogonal direction of research is that of identifying domains for which a stable and strategyproof mechanism \emph{does} exist. This can be done by placing \emph{restrictions on the admissible preferences}~\citep{AlBa94a,Akah14a,Koji07a,Kest11a}. 
DAA has been shown to be \emph{strategyproof} for the proposing-side if all proposing agents have capacity one \citep{DuFr81a,RoSo90a}.
Proposees in general can however potentially manipulate the outcome of DAA and the same holds for proposing agents with capacity greater than one, \eg colleges in college admission markets.



There has also been work on algorithms to compute optimal preference reports. 
\citet{TST01a} consider manipulation by an individual women $w$ in one-to-one markets when all other agents state their preferences truthfully and show that a beneficial manipulation by $w$, if one exists, can be found in polynomial time both in the case in which women can truncate their preference list as well as in the case in which they need to submit complete preference lists.

The computational complexity problem of computing a manipulation has already been studied in great depth for voting rules~\citep[see \eg][]{FHH10a,FaPr10a}.
For one-to-one matching markets, \citet{KoMa10b} and \citet{SuYa11a} considered the complexity of computing beneficial manipulation for \emph{groups} of agents.
In general agents may be allowed to express some lesser preferred agents as \emph{unacceptable}. Stability of matching requires that no agent is matched to an unacceptable agent. An agent can implicitly express some agent $j$ as unacceptable by not including $j$ in its preference list.
In view of this allowance, an agent may misreport his preference by \emph{blacklisting/dropping} an agent but not blacklisting a lesser preferred agent or by \emph{truncating} his preference list by not including some least preferred acceptable agents in the list. \citet{JKK12a} have recently shown that \emph{truncation} of true preferences is an exhaustive manipulation strategy for agents with capacity one. This result implies that the problem of finding a beneficial manipulation for a given student $s$ under college-proposing DAA is solvable in polynomial time by checking 
each truncated preference list.  
 In many real-world matching markets agents are required to submit complete preference lists, which prevents manipulation via blacklisting or truncation~\citep{TST01a}. This is a reasonable assumption as seen in the context of  college admissions: as long as all the applying student satisfy some minimum requirements, they cannot be deemed unacceptable by any college. Colleges may want each slot should be filled, preferably by good students, but by any student if necessary. 
 If agents are not allowed to express truncated preference lists, it effectively means that agents are not allowed to express any agent as unacceptable. 
 Agents may still manipulate in these markets by submitting preference lists that are complete but do not reflect their true preferences. We call manipulation via submitting such a falsified preference list \emph{misrepresenting preferences}.
 Even if agents are aware of the preferences of the other agents, misrepresenting preferences at least guarantees the agent a match whereas dropping or truncation strategies by the agent may leave the agent empty-handed in uncertain scenarios.
 In many-to-one and many-to-many markets the outcome of a stable matching procedure does not only depend on the submitted  preferences, but also on the reported capacities. It was first shown by \citet{Sonm97a} that \emph{underreporting capacities}, i.e. claiming a capacity $q'<q$ where $q$ is the capacity of a college, can lead to a beneficial manipulation. 
 When indifferences in students preferences are allowed, colleges can also benefit from \emph{overreporting capacities}, i.e. claiming a capacity $q'>q$ where $q$ is the capacity of a college.


While previous work is focussed on manipulation by agents with capacity one, we investigate the computational complexity of finding a beneficial manipulation for a given \emph{college} in college admission markets. 

 \section{Model}
\label{section:preliminaries}

Two-sided matching markets contain two disjoint sets of agents~$A$ and~$B$. Every agent~$i$ has preferences~$>_i$ over the agents on the other side of the market. If $a >_i a'$ then agent $i$ prefers being matched with $a$ over being matched with $a'$. We call the set~$>=(>_i)_{i\in A\cup B}$ consisting of the preference relations of all agents a \emph{preference profile}. Throughout this paper we assume all agents to have \emph{strict, complete, and transitive preferences} unless indicated otherwise. 
The goal in these markets is to find a mapping~$\mu$ from agents of one set to agents of the other set, called a \emph{matching}, with $a\in \mu(b) \Leftrightarrow b\in \mu(a)$. 

Famous examples of two-sided matching markets are the marriage problem and the college admission problem as first introduced by \citet{GaSh62a}. 
In the marriage problem each man/woman is matched with at most one woman/man. In such a case, we speak of a \emph{one-to-one} matching market.
In a market $\left(C,S,q,>\right)$ of the college admission type, which we focus on in this paper, we have a set~$C$ of colleges and a set~$S$ of students. Generic representatives of these groups are denoted by~$c$ and~$s$, respectively. 
Every student may only be matched to one college whereas every college~$c$ is equipped with a capacity (or quota) value~$q(c)\geq 1$ denoting the number of students it may be matched with. Such a market is referred to as a \emph{many-to-one} matching market.
There also exist markets for \emph{many-to-many matching}, \eg for matching firms and workers, each of which may work for several firms. However, these markets will not be considered in detail in this paper.

When dealing with many-to-one (or many-to-many) matching markets, we need to extend the preference relation~$>_c$ of a college~$c$ to preferences over \emph{sets} of students in a meaningful way. To this end, we assume \emph{responsive preferences} throughout this paper: for all $S' \subset S$ with $|S'| < q(c)$ and $s,s' \in S\setminus S'$ it holds that
\begin{align*}
	S' \cup \{s\} >_c& S' \cup \{s'\} \text{ if and only if $s >_c s'$ and}\\
	S' \cup \{s\} >_c& S'.
\end{align*}
Note that a college $c$ only prefers adding an additional student to its set of matches over not doing so if its capacity $q(c)$ is not yet filled. Any matching that assigns more than $q(c)$ students to a college $c$ is not valid.
Thus each agent expresses preferences over sets of agents by \emph{only} expressing preferences over agents.
When there is incomparability between two sets of allocations with respect to responsiveness, we will assume that the agent is indifferent between them. Alternatively, even if an agent is not indifferent among such sets, 
when we consider \emph{beneficial manipulations}, we will only consider those in which the result of manipulation dominates the original allocation with respect to the responsive relation. For example, if the  preferences are $a>b>c>d$ and if $\{b,c\}$ is the match, then $\{a,d\}$ cannot be the outcome of a beneficial manipulation since $\{a,d\}$ is not better outcome than $\{b,c\}$ with respect to the responsive relation.


\section{Deferred Acceptance\\Algorithms}
\label{section:DAA}
The \emph{deferred acceptance algorithm} (DAA) was first officially introduced in 1962 by \citeauthor{GaSh62a}.\footnote{\citeauthor{Roth84b} reported later that DAA has already been in use by the National Resident Matching Program (NRMP) since 1951 for matching medical students to hospitals for their medical internship or residency in the United States. There, DAA naturally evolved through a trial and error process that is described in \citep{Roth84b}.}
Today, DAA is applied on numerous real world problems (see e.g., http://www.matching-in-practice.eu).

An execution of DAA consists of several rounds. In each round, agents from one side of the market, the \emph{proposers}, propose to agents of the other side, the \emph{proposees}, to be matched with them. Proposees can tentatively accept a proposal, but may revoke their acceptance should they in a later round receive a proposal from an agent they prefer to their tentative match. Their final decisions regarding their matching partners are therefore \emph{deferred} until the final round.

Depending on whether students or colleges are considered the ``proposing'' side during the algorithm, there are two variants of DAA for the college admission market, \ie \emph{\studentproposing} and \emph{\collegeproposing} DAA.

\paragraph{Student-proposing DAA}
\begin{enumerate}
\item Each student applies to his favorite college.
\item Each college rejects all applications from students that are unacceptable to it. 
If a college received at most $q$ applications from acceptable students so far, all those students are put on the colleges waiting list. Otherwise the college puts its favorite $q$ students among all applicants on the waiting list and rejects all remaining ones. 
\item Each student that was rejected in the previous step applies to his favorite among the colleges he or she has not yet applied to.
\item Steps 2 and 3 are repeated until it holds for all students that they were either not rejected in the previous step or already applied to all colleges acceptable to them.
\item Each college admits all students on its waiting list.
\end{enumerate}


\normalsize
By the use of the waiting list, all final admissions are deferred to the end of the procedure. 
The algorithm can easily be adapted to colleges making offers to students.

\paragraph{College-proposing DAA}
\begin{enumerate}
\item Each college makes offers to its $q$ favorite students.
\item Each student keeps the offer of his or her favorite acceptable college among those that made an offer to him/her so far (if such a college exists) and rejects all others.
\item Each college that was rejected in the previous step makes offers to its $k$ favorite among the students it has not yet made offers to, where $k$ is the number of students that rejected that college in the previous step. If there are less than $k$ students left that the college did not yet make offers to, it makes offers to all those students.
\item Steps 2 and 3 are repeated until it holds for all colleges that they were either not rejected in the previous step or already made offers to all students acceptable to them.
\item Each student gets admitted by the college whose offer he kept.
\end{enumerate}
\normalsize

A result regarding DAA in general that is frequently used when arguing about manipulation is that, if preferences are strict, the order in which proposers are considered, \ie make their proposals, does not affect the outcome of DAA \citep{MaGo89a}. This means that given strict preferences, each variant of DAA has a unique outcome.
In many real-world matching markets agents are required to submit complete preference lists, which prevents manipulation via blacklisting or truncation. This effectively means that agents are not allowed to express any agent as unacceptable. 
Agents may still manipulate in these markets by submitting preference lists that are complete but do not reflect their true preferences. We call manipulation via submitting such a falsified preference list \emph{misrepresenting preferences}.
Even if agents are aware of the preferences of the other agents, misrepresenting preferences at least guarantees the agent a match whereas dropping or truncation strategies by the agent may leave the agent empty-handed in uncertain scenarios.

\section{Manipulation of student-proposing DAA}
 \label{section:o2mComplexity}
We show that it can be checked in polynomial time whether a given college $c \in C$ can misrepresent its preferences and obtain a better outcome. 
Let $S, C$ be the agents of a many-to-one matching market and $\mu$ denote the matching that is obtained by applying \studentproposing DAA to this market. First it is shown in that whenever the market admits a beneficial manipulation $\mu'$ by a college $c \in C$, then there also exists a beneficial manipulation $\mu''$ by $c$ s.t. $|\mu(c) \setminus \mu''(c)|=1$. We then use this fact to reduce the problem to that of finding a beneficial manipulation via misrepresentation for a given proposee in a one-to-one matching market, which was shown by \citet{TST01a} to be solvable in polynomial time.
We describe some key concepts that we will use for our algorithmic result.
\paragraph{New Proposals}
In order for a college $c$ to potentially benefit from a manipulation, it must during the manipulated execution of DAA receive proposals from students that do not propose to it during the execution of DAA under truthful preferences and that it prefers to some of its matches under $\mu$. Since all agents but $c$ are assumed to state their preferences truthfully, any such \emph{new proposals} must somehow be caused by $c$ misrepresenting its preferences. We formally define a new proposal as follows:
Let $s \in S, c \in C$ and $\mu'$ be a manipulation s.t. $s$ proposes to $c$ during the manipulated execution of DAA leading to $\mu'$. We call the proposal from $s$ to $c$ a \emph{new proposal} if $s$ does not propose to $c$ during the execution of DAA under truthful preferences. 

\paragraph{Rejection in favor}
The only way a misrepresentation of preferences by $c$ can cause new proposals, is if it leads to $c$ rejecting some student(s) $t$ it does not reject during the truthful run of DAA. Since the only students that $c$ receives proposals from during the truthful run of DAA and does not reject are those it is matched with, we obtain $t \in \mu(c)$. Therefore, after being rejected by $c$, $t$ will make one or more new proposals to colleges $c'$ s.t. $c >_t c'$, if such colleges exist. 

If $t$ is accepted by some such $c'$, this college might reject another student $s'$ \emph{in favor} of $t$.
Formally, let $s', t \in S$ and $c' \in C$. We say that \emph{$c'$ rejects $s'$ in favor of $t$} if both $s'$ and $t$ propose to $c'$ during the execution of student-proposing DAA and one of the following holds: 
	\begin{enumerate}
	\item At the time at which $t$ proposes to $c'$, $c'$ has $q(c')$ temporary matches. Student $s'$ is the least preferred, according to $>_{c'}$, among those temporary matches and $t >_{c'} s'$. Therefore $c'$ rejects $s'$ in order to be able to accept $t$.
	\item At the time at which $s'$ proposes to $c'$, $c'$ has $q(c')$ temporary matches and rejects $s'$. Student $t$ is among those temporary matches. If $s'$ had proposed to $c'$ before $t$, 1. would have applied. Therefore $c'$ would not have rejected $s'$ if $t$ had not proposed to it and everything else stayed the same. 
	\end{enumerate}

\paragraph{Important Sets of Students}
A student never proposes to the same college twice. Therefore it holds for any student $t \in \mu(c)$ that $c$ rejects in order to manipulate DAA, that $c$ is not matched with $t$ under $\mu'$. 
Let $T= \mu(c) \setminus \mu'(c)$ denote the set of students that $c$ is matched with under $\mu$ and rejects in order to manipulate the outcome of DAA. 
Let $S^*= \mu'(c) \setminus \mu(c)$ denote the set of students that $c$ is matched with under $\mu'$, but not under $\mu$. 
According to definition of responsive preferences, in order for $\mu'$ to be a beneficial manipulation for $c$ w.r.t. $\mu$, it must hold that for each $t \in T$ there exists a distinct $s \in S^*$ s.t. $s >_c t$. We therefore obtain $|T|=|S^*|$.
Under responsive preferences, a college $c$ does not reject any proposals from acceptable students $s$, i.e. $s >_c \emptyset$, as long as its capacity $q(c)$ is not yet filled. Therefore if $c$ rejects an acceptable student $s$, it must have received $q(c)$ proposals from students $s'$ s.t. $s' >_c s$, i.e. an acceptable student may only be rejected in favor of another student. Note that this implies that only colleges $c$ with $|\mu(c)|=q(c)$ who receive more than $q(c)$ proposals during the truthful run of DAA can manipulate the outcome of student-proposing DAA via misrepresenting their preferences, as formally stated in the following lemma. 

\begin{lemma}\label{th:collegeswithhighq}
Let $S, C$ be the agents of a many-to-one matching market in which all colleges are required to submit complete preference lists and $\mu$ the matching obtained by applying \studentproposing DAA to this market. Given responsive preferences over sets, only colleges $c \in C$ s.t. $|\mu(c)|=q(c)$ who receive more than $q(c)$ proposals during the truthful run of DAA can manipulate the outcome of student-proposing DAA via misrepresenting preferences.
\end{lemma}
\begin{proof}
Assume that $c \in C$ can manipulate DAA via misrepresenting its preferences while all other agents state their preferences truthfully. Then during the manipulated execution of DAA, $c$ must accept a student it does not accept during the truthful run of DAA  or reject a student it does not reject during the truthful run of DAA.

Since all students are effectively acceptable to $c$, under responsive preferences $c$ \emph{must} accept all proposals from students as long as its capacity is not yet filled. Therefore if $|\mu(c)|<q(c)$ and/or $c$ receives $\leq q(c)$ proposals during the truthful run of DAA, $c$ may not reject any of these proposals and there exist no students that $c$ receives a proposal from and does not accept. Since by assumption $c$ can manipulate the outcome of DAA, we obtain that $|\mu(c)|=q(c)$ and $c$ receives more than $q(c)$ proposals during the execution of DAA under truthful preferences.
\end{proof}

In order to be able to reject a student $t \in T$, $c$ must accept $q(c)$ proposals from students $s \neq t$. Since any new proposals only occur after $c$ rejected some $t \in T$, this implies that $c$ must accept at least one proposal from a student $u$ s.t. $u \not \in \mu(c)$ and $u$ also proposes to $c$ during the truthful run of DAA. Therefore $\forall t \in T: u <_c t$ according to the truthful preferences. 

Let $U$ denote the set of students $u$ that $c$ temporarily accepts in order to reject students $t \in T$ in favor of $u$.
Since $c$ can only reject a single student $t \in T$ in favor of each student $u \in U$, we obtain $|U|=|T|$ and thereby also $|U|=|S^*|$. Due to students $u$ being less preferred than students $t$ according to the truthful preferences, $c$ may not be matched with any $u \in U$ in a beneficial manipulation: 
$U \cap \mu'(c) = \emptyset.$
Therefore $c$ must reject all students $u \in U$ again upon accepting proposals from students $s \in S^*$. Since $|U|=|S^*|$ this means that whenever $c$ receives a new proposal from a student $s \in S^*$, a student $u \in U$ is rejected in its favor.

\paragraph{Chain of Rejections}
Since $c$ is matched with no student $s \in S^*$ under $\mu$, but it holds for all those students that there exists at least one student $t \in \mu(c)$ that $s$ is preferred to, no student $s \in S^*$ proposes to $c$ during the truthful run of DAA. Therefore all proposals from some $s \in S^*$ to $c$ that are made during the manipulated run of DAA are new proposals.

As stated above, since $c$ is assumed to be the only manipulating college, all new proposals must, directly or indirectly, be caused by $c$ rejecting students $t \in T$. Obviously the rejection of a student $t \in T$ can not immediately cause any new proposals to $c$. Upon being rejected by $c$, $t$ will propose to colleges $c' <_t c$ in order of his or her preferences until $t$ is either accepted by some such $c'$ or proposed to all colleges.

If $t$ is accepted by such a college $c'$, a student $s'$ might be rejected in favor of $t$. According to definition of `rejecting in favor', 
$c'$ does not reject $s'$ if $t$ does not propose to it while everything else stays the same. Since $t$ does not propose to $c'$ during the truthful run of DAA, we obtain that $s'$ is not rejected by $c'$ during the truthful run. Therefore any proposals made by $s'$ after being rejected by $c'$ are new proposals. 

Should it hold that $s' \in S^*$ and $s'$ is not accepted by any $c''$ s.t. $c' >_{s'} c'' >_{s'} c$, then the new proposal by $s'$ to $c$ is caused by: 
\begin{enumerate}
	\item $c$ rejecting $t$, who then proposes to $c'$
	\item $c'$ rejecting $s'$ in favor of $t$, who then proposes to $c$
	\end{enumerate}
We say that $c$ rejecting $t$ triggered a \emph{chain of rejections} leading to $s'$ proposing to $c$.
Clearly such a chain of rejections could include more steps. Student $s'$ might instead be accepted by some college $c''$ who rejects another student $s''$ in its favor and so on.

%
\paragraph{Illustrative Example}
Before we proceed with our argument, consider the following example of a many-to-one matching market with sets
\begin{align*}
S =& \{s_1,s_2,s_3,s_4,t_1,t_2,t_3,u_1,u_2,u_3\}\\
C =&\{c,c_1,c_2,c_3,c_4\}
\end{align*}
and the preference profile depicted in \figref{fig:example-profile}, where ``$?$'' indicates that the preferences over the agents that are not mentioned explicitly are irrelevant.

\begin{figure}[tb]
\center
\scalebox{.9}{
\begin{tabular}{c c c c c c c c c c}
$s_1$ & $s_2$ & $s_3$ & $s_4$ & $t_1$ & $t_2$ & $t_3$ & $u_1$ & $u_2$ & $u_3$\\
\midrule
$c_1$ & $c_2$ & $c_3$ & $c_4$ & $c$ & $c$ & $c$ & $?$ & $?$ & $?$\\
$c_2$ & $c_3$ & $c$ & $c$ & $c_1$ & $c_2$ & $c_3$ & $?$ & $?$ & $?$\\
$?$ & $c$ & $?$ & $?$ & $?$ & $c_4$ & $?$ & $?$ & $?$ & $?$\\
$?$ & $?$ & $?$ & $?$ & $?$ & $?$ & $?$ & $?$ & $?$ & $?$\\
$?$ & $?$ & $?$ & $?$ & $?$ & $?$ & $?$ & $c$ & $c$ & $c$\\
& & & & & & & & &\\
& & & & & & & & &\\
& & & & & & & & &\\
& & & & & & & & &\\
& & & & & & & & &
\end{tabular}%
\hspace{.5cm}%
\begin{tabular}{c c c c c}
$c$ & $c_1$ & $c_2$ & $c_3$ & $c_4$\\
\midrule
$s_4$ & $t_1$ & $s_1$ & $t_3$ & $t_2$\\
$t_3$ & $s_1$ & $t_2$ & $s_2$ & $s_4$\\
$s_2$ & $?$ & $s_2$ & $s_3$ & $?$\\
$t_1$ & $?$ & $?$ & $?$ & $?$\\
$s_3$ & $?$ & $?$ & $?$ & $?$\\
$t_2$ & $?$ & $?$ & $?$ & $?$\\
$s_1$ & $?$ & $?$ & $?$ & $?$\\
$u_1$ & $u_1$ & $u_1$ & $u_1$ & $u_1$\\
$u_2$ & $u_2$ & $u_2$ & $u_2$ & $u_2$\\
$u_3$ & $u_3$ & $u_3$ & $u_3$ & $u_3$
\end{tabular}
}
\label{fig:example-profile}
\caption{Example of a preference profile in a many-to-one matching market.}
\end{figure}

Given capacities $q(c)=3$ and $q(c_i)=1$ for all $i \in \{1,...,4\}$ \studentproposing DAA results in
\[\mu = \left(\begin{matrix}s_1 & s_2 & s_3 & s_4 & t_1 & t_2 & t_3 & u_1 & u_2 & u_3\\c_1 & c_2 & c_3 & c_4 & c & c & c & \emptyset & \emptyset & \emptyset\end{matrix}\right).\]
If $c$ misrepresents its preferences as
\[s_4 >_c s_2 >_c s_3 >_c u_1 >_c u_2 >_c u_3 >_c s_1 >_c t_3 >_c t_1 >_c t_2\]
\studentproposing DAA instead gives
\[\mu' = \left(\begin{matrix}s_1 & s_2 & s_3 & s_4 & t_1 & t_2 & t_3 & u_1 & u_2 & u_3\\c_2 & c & c & c & c_1 & c_4 & c_3 & \emptyset & \emptyset & \emptyset\end{matrix}\right).\]
All colleges strictly prefer $\mu'$ to $\mu$, while $u_1$, $u_2$ and $u_3$ are indifferent between the two matchings and all remaining students strictly prefer $\mu$ to $\mu'$.

In this example we have $T=\{t_1,t_2,s_3\}, S^*=\{s_2,s_3,s_4\}$, and $U=\{u_1,u_2,u_3\}$.

All proposals by $u_1, u_2, u_3$ as well as proposals from all remaining students to their most preferred college occur during both the truthful and the manipulated execution of DAA. Any additional proposals made during the manipulated run of DAA are new proposals. Those new proposals are illustrated in \figref{fig:rejectionchains}.
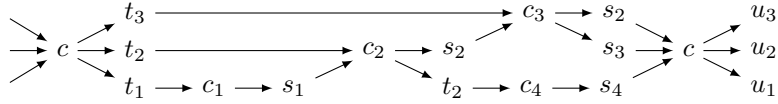
\begin{figure}[tb]
	\def\nd{.5cm}
\label{fig:rejectionchains}
\center
			\scalebox{1}{
\begin{tikzpicture}[node distance=\nd]

\node[draw=none] (in2) at (0,0) {};
\node[draw=none, yshift=-\nd] (in1) at (in2) {};
\node[draw=none, yshift= \nd]  (in3) at (in2) {};
\node[right=of in2			 ] (c)  {$c$};
\node[right=of c ,yshift=-\nd] (t1) {$t_1$};
\node[right=of c ,			 ] (t2) {$t_2$};
\node[right=of c ,yshift= \nd] (t3) {$t_3$};
\node[right=of t1,			 ] (c1) {$c_1$};
\node[right=of c1,			 ] (s1) {$s_1$};
\node[right=of s1,yshift= \nd] (c2) {$c_2$};
\node[right=of c2,			 ] (s2) {$s_2$};
\node[right=of c2,yshift=-\nd] (t22) {$t_2$};
\node[right=of s2,yshift= \nd] (c3) {$c_3$};
\node[right=of t22,			 ] (c4) {$c_4$};
\node[right=of c4,			 ] (s4) {$s_4$};
\node[right=of c3,yshift=-\nd] (s3) {$s_3$};
\node[right=of c3,			 ] (s22) {$s_2$};
\node[right=of s3,			 ] (cc) {$c$};
\node[right=of cc,			 ] (u2) {$u_2$};
\node[right=of cc,yshift=-\nd] (u1) {$u_1$};
\node[right=of cc,yshift= \nd] (u3) {$u_3$};
\draw[-latex] (in1) -- (c);
\draw[-latex] (in2) -- (c);
\draw[-latex] (in3) -- (c);
\draw[-latex] (c) -- (t1);
\draw[-latex] (c) -- (t2);
\draw[-latex] (c) -- (t3);
\draw[-latex] (t1) -- (c1);
\draw[-latex] (c1) -- (s1);
\draw[-latex] (s1) -- (c2);
\draw[-latex] (t2) -- (c2);
\draw[-latex] (c2) -- (t22);
\draw[-latex] (c2) -- (s2);
\draw[-latex] (s2) -- (c3);
\draw[-latex] (t3) -- (c3);
\draw[-latex] (t22) -- (c4);
\draw[-latex] (c3) -- (s22);
\draw[-latex] (c3) -- (s3);
\draw[-latex] (c4) -- (s4);
\draw[-latex] (s22) -- (cc);
\draw[-latex] (s3) -- (cc);
\draw[-latex] (s4) -- (cc);
\draw[-latex] (cc) -- (u1);
\draw[-latex] (cc) -- (u2);
\draw[-latex] (cc) -- (u3);
\end{tikzpicture}
}
\caption{Illustration of chains of rejections in the manipulation example. An edge from $s \in S$ to $c \in C$ indicates $s$ proposing to $c$, an edge from $c \in C$ to $s \in S$ indicates $c$ rejecting $s$ in favor of an incoming proposal.}
\end{figure}
The graph contains three paths representing chains of rejections triggered by the manipulating college $c$ rejecting some $t \in T$ and leading to some $s \in S^*$ proposing to  $c$:
\begin{enumerate}
\item $c \rightarrow t_1 \rightarrow c_1 \rightarrow s_1 \rightarrow c_2 \rightarrow t_2 \rightarrow c_4 \rightarrow s_4$
\item $c \rightarrow t_2 \rightarrow c_2 \rightarrow s_2 \rightarrow c_3 \rightarrow s_3$
\item $c \rightarrow t_3 \rightarrow c_3 \rightarrow s_2$
\end{enumerate}

\paragraph{Beneficial Manipulation by Rejecting a Single Match}
\label{section:rejectingoneenough}
We show that each chain of rejections triggered by the manipulating college $c$ rejecting some $t \in T$ leads to a distinct $s \in S^*$ proposing to $c$. Moreover, it holds for all $t \in T$ that there exists a manipulation that is achieved via $c$ misrepresenting its preferences while all other agents state their preferences truthfully, s.t. $T'=\{t\}$ and some $s \in S^*$ proposes to $c$ during the manipulated run of DAA. Using those results, we then show that whenever \collegeproposing DAA admits a beneficial manipulation by a college $c$, then it also admits a beneficial manipulation by $c$ s.t. $|T'|=1$.
\begin{lemma}
\label{eachtones}
Let $c \in C$ s.t. $c$ can beneficially manipulate DAA. Then it holds for each $t \in T$ that the rejection of $t$ by $c$ triggers exactly one chain of rejections leading to a distinct $s \in S^*$ proposing to $c$.
\end{lemma}
\begin{proof}
According to definition of in favor, 
a college rejects at most one student in favor of another student proposing to it. The rejection of some $t \in T$ by $c$ can therefore trigger at most one chain of rejections. Once such a chain leads to some $s \in S^*$ proposing to $c$, some $u \in U$ is rejected in favor of $s$. Since $c$ also rejected all students $u \in U$ during the truthful run of DAA, the rejection of $u$ by $c$ during the manipulated run of DAA will by itself not lead to any new proposals to $c$.

After being rejected by $c$, $u$ will propose to colleges less preferred than $c$, if such colleges exist. Since all colleges $c'$ s.t. $c >_u c' >_u \mu(u)$ reject $u$ during the truthful run of DAA and all colleges but $c$ state their preferences truthfully, $u$ will also be rejected by all those colleges $c'$ during the manipulated execution of DAA. Therefore $u$ will propose to $\mu(u)$.

If $\mu(u)=\emptyset$ or $\mu(u)$ does not reject either $u$ or another student $v \in \mu(\mu(u))$ in favor of $u$, then no new proposals are triggered. Since $\mu(u)$ states its preferences truthfully, new proposals can therefore only be triggered if $\mu(u)$ previously received a new proposal itself, namely from a student $s'$ it prefers to at least one of its matches under $\mu$.

Such a new proposal can only have been caused by a chain of rejections triggered by $c$ rejecting some $t \in T$. Prior to being proposed to by $u$ and receiving any new proposals, $\mu(u)$ is either matched with $< q(\mu(u))$ students or the least preferred among its temporary matches is a student that $\mu(u)$ also rejects during the truthful run of DAA. Therefore $s'$ proposing to $\mu(u)$ does not immediately cause any new proposals.

In sum, $c$ rejecting $u$ can only lead to new proposals if a chain of rejections triggered by $c$ rejecting some $t \in T$ lead to the new proposal by $s'$ to $\mu(u)$. Then the new proposals caused by $c$ rejecting $u$ were initially triggered by $c$ rejecting $t$.

Therefore each chain of rejections triggered by $c$ rejecting some $t \in T$ can lead to at most one $s \in S^*$ proposing to $c$. Since as stated above each $t \in T$ can trigger at most one such chain and $|T|=|S^*|$, we obtain that each $t \in T$ triggers exactly one chain of rejections leading to a distinct $s \in S^*$ proposing to $c$.
\end{proof}

 Using this insight, we now show that it holds for each $t \in T$ that if $c$ misrepresents its preferences s.t. $T'=\{t\}$, some $s \in S^*$ proposes to $c$ during the such manipulated execution of DAA. Note that this does not yet imply that $c$ can \emph{beneficially} manipulate DAA via such a misrepresentation, since it does not necessarily hold that $s >_c t$.

\begin{lemma}
\label{lem:anytleadstosomes}
Let $c \in C$ s.t. $c$ can beneficially manipulate DAA, $t \in T$ and $\mu''$ be a manipulation resulting from $c$ misrepresenting its preferences s.t. $T'=\{t\}$. Then some $s \in S^*$ proposes to $c$ during the manipulated execution of DAA yielding $\mu''$.
\end{lemma}
\begin{proof}
According to Lemma \ref{eachtones}, the rejection of $t$ by $c$ triggers a chain of rejections leading to a distinct $s \in S^*$ proposing to $c$. If that chain is not affected by any other $t' \in T \setminus \{t\}$, then clearly $c$ can misrepresent its preferences such that $T'=\{t\}$, leading to $s$ proposing to $c$. Unaffected means that the students that are rejected along that chain are $\not \in T$ and did not propose to the college they are rejected from as part of a chain leading from some $t' \in T \setminus \{t\}$ to some $s \in S^*$ proposing to $c$.

We consider the two ways in which the chain of rejections triggered by $c$ rejecting $t$ can be affected by some $t' \in T \setminus \{t\}$ separately.
\begin{itemize}[noitemsep,topsep=0pt,parsep=0pt,partopsep=0pt]
\item[1.] A student rejected along that chain is some $t' \in T$. \footnote{The first chain of rejections in the example is an instance of this case.}

Let $c'$ be the college that rejects $t'$ and $s'$ be the student in whose favor $c'$ rejects $t'$. Then $c'$ was at some point proposed to by $t'$ and rejected a student $s''$ in its favor. Since $s' >_{c'} t'$ and $t' >_{c'} s''$ we obtain by transitivity of preferences that $s' >_{c'} s''$.

Therefore if $t'$ does not propose to $c'$, but $s'$ does, $c'$ will reject $s''$ in its favor. Since $c'$ rejecting $s''$ in favor of $t'$ was part of the path from $t'$ to some $s \in S^*$, the new path containing $c'$ rejecting $s''$ for $s'$ will also end in this same $s \in S^*$.
\item[2.] A student $s'$ that is rejected by some $c'$ along the chain only proposed to $c'$ due to a chain of rejections triggered by some $c$ rejecting some $t' \in T \setminus \{t\}$. 

Let $s''$ be the student that $c'$ rejected in favor of $s'$ and $s'''$ be the student in whose favor $c'$ rejected $s'$. Then the chain of rejections that is triggered by $c$ rejecting $t'$ and includes $s'$ proposing to $c'$ which rejects $s''$ in its favor leads to some $s \in S^*$ proposing to $c$. Since $s''' >_{c'} s'$ and $s' >_{c'} s''$ we obtain
by transitivity of preferences that $s''' >_{c'} s''$.

Therefore if $s'$ does not propose to $c'$, but $s'''$ does, $c'$ will reject $s''$ in its favor. Since $c'$ rejecting $s''$ in favor of $s'$ was part of the path from $t'$ to some $s \in S^*$, the new path containing $c'$ rejecting $s''$ for $s'''$ will also end in this same $s \in S^*$.
\end{itemize}

Therefore it holds for each $t \in T$ that $c$ can misrepresent its preferences s.t. $T'=\{t\}$ and some $s \in S^*$ proposes to $c$ during the manipulated execution of DAA, though not necessarily s.t. $s >_c t$.
\end{proof}

Using \lemref{lem:anytleadstosomes}, we prove the following.

\begin{theorem}
\label{rejectJustOne}
If a given college $c$ can beneficially manipulate student-proposing DAA via misrepresenting preferences, then there also exists a beneficial manipulation by $c$ s.t. $|T'|=1$.
\end{theorem}
\begin{proof}
Let $t_k$ be the least preferred student among $T$, i.e. $\forall t' \in T \setminus \{t_k\}: t' >_c t_k$. Then according to Lemma \ref{lem:anytleadstosomes} there exist a misrepresentation of preferences by $c$ and a student $s \in S^*$ s.t. $T'=\{t_k\}$ and $s$ proposes to $c$ during the manipulated execution of DAA yielding $\mu''$.
Since \emph{all} students $s \in S^*$ are strictly preferred to $t_k$,$c$ will accept this proposal and thereby  $\mu''(c)=\{s\} \cup \mu(c) \setminus \{t_k\}$. The manipulation $\mu''$ is therefore beneficial for $c$ w.r.t. $\mu$.
\end{proof}


\paragraph{Reduction to One-to-One}
\label{section:reduction}
Using Theorem \ref{rejectJustOne}, we show that the problem of deciding whether a given college $c \in C$ can beneficially manipulate student-proposing DAA via misrepresenting its preferences when all other agents state their preferences truthfully can be reduced to that of finding a beneficial manipulation via misrepresentation for a given proposee in a one-to-one matching market.
We first show that the execution of \studentproposing DAA for a many-to-one matching market can be simulated via executing DAA for an associated one-to-one market as suggested in \citep[][p. 133ff]{RoSo90a}:
\begin{lemma}
\label{one2oneReduction}
Let $C, S$ be the agents of a many-to-one matching market, $\mathcal{P}$ their preference profile and $\mu_\text{m2o}$ denote the matching that results from applying \studentproposing DAA to the market. Then there exists a one-to-one matching market with sets of agents $C', S$ and preference profile $\mathcal{P}'$ and a mapping function $f: C \rightarrow C'$ s.t. 
$\forall c \in C: \mu_\text{m2o}(c) = \bigcup_{c_i \in f(c)} \mu_\text{o2o}(c_i)$
where $\mu_\text{o2o}$ denotes the matching that results from applying DAA to the one-to-one market, and
$\forall s \in S: f(\mu_\text{m2o}(s)) \ni \mu_\text{o2o}.$
\end{lemma}

 By combining Theorem \ref{rejectJustOne} and Lemma \ref{one2oneReduction}, we reduce the problem of deciding whether a given college $c$ can beneficially manipulate DAA to that of deciding whether at least one agent $c_i \in f(c)$ can beneficially manipulate DAA in the related one-to-one market. This is formally stated in the following theorem.

\begin{theorem}
\label{iffonecanmanipulate}
Let $C, S$ be the agents of a many-to-one matching market, $\mathcal{P}$ their preference profile and $c \in C$. Then $c$ can beneficially manipulate \studentproposing DAA via misrepresenting its preferences if and only if in the associated one-to-one market with sets of agents $C',S$, preference profile $\mathcal{P}'$ and mapping function $f$ there exists at least one $c_i \in f(c)$ s.t. $c_i$ can beneficially manipulate DAA via misrepresenting its preferences. 
\end{theorem}

 The problem of finding a beneficial manipulation via misrepresentation for a given proposee in one-to-one matchings was shown to be solvable in polynomial time in \citep{TST01a}. We therefore obtain our final result.
\begin{theorem}
\label{thm:o2mInP}
The problem of deciding whether a given college $c \in C$ can beneficially manipulate \studentproposing DAA via misrepresenting its preferences is solvable in polynomial time.
\end{theorem}
\begin{proof}
By Theorem \ref{iffonecanmanipulate}, we can decide the problem by checking whether in the associated one-to-one market at least one $c_i \in f(c)$ can beneficially manipulate DAA via misrepresenting its preferences. Since $|f(c)|=q(c)$ and a possible beneficial manipulation for a given proposee in a one-to-one market, can be found in polynomial time, the statement follows. 
\end{proof}

Regarding the problem of finding an \emph{optimal} manipulation for a college under \collegeproposing DAA, we immediately get a polynomial algorithm by iteratively applying the steps from the proof of \thmref{thm:o2mInP}. As long as there still is a beneficial manipulation, we find it in polynomial time and the number of possible improvements for a single college is polynomial in the input size.

\begin{corollary}
A college's optimal manipulation by misrepresenting preferences under \studentproposing DAA can be found in polynomial time.
\end{corollary}
Note, however, that there may be many optimal strategies as many matchings are incomparable under responsive preferences.

\section{Manipulation of college-proposing DAA}
\label{section:m2oComplexity}
We show that the problem of deciding whether a given college $c \in C$ can beneficially manipulate college-proposing DAA via misrepresenting its preferences when all other agents state their preferences truthfully is fixed-parameter tractable. The worst case complexity of finding a beneficial manipulation for $c$, if one exists, depends only on its capacity $q(c)$, not on the size of the market. 
Let $S, C$ be the agents of a many-to-one matching market and $\mu$ denote the matching that is obtained by applying DAA to this market. In general, a college $c \in C$ can manipulate the outcome of DAA by accepting a student it does not accept during an execution of DAA under truthful preferences or by rejecting a student it is matched with under $\mu$. Under \collegeproposing DAA, this can be achieved by making new proposals 
or by not proposing to a student the college is matched with under $\mu$. Making new proposals can never lead to a beneficial manipulation, as formally stated in the following lemma. 

\begin{lemma}
\label{npnotbeneficial}
Let $S, C$ be the agents of a many-to-one matching market and $\mu$ the matching that is obtained by applying DAA to the market. Given responsive preferences over sets, no college can benefit from a manipulation in which it proposes to a student it does not propose to during an execution of college-proposing DAA under the truthful preferences.
\end{lemma}

 A beneficial manipulation can therefore only be achieved by falsifying preferences s.t. $c$ does not propose to at least one student $s \in \mu(c)$ during the manipulated execution of DAA. When considering manipulation via misrepresenting preferences, this implies that only colleges that fill their capacity under $\mu$ can beneficially manipulate DAA. Using Lemma \ref{npnotbeneficial}, we can prove the following lemma.
\begin{lemma}
\label{onlyfullcapacity}
Let $S, C$ be the agents of a many-to-one matching market in which all agents are required to submit complete preference lists, $\mu$ the matching resulting from applying DAA to the market and $c \in C$ s.t. $c$ can beneficially manipulate DAA via misrepresenting its preferences. Then $|\mu(c)|=q(c)$.
\end{lemma}
\begin{proof}
According to Lemma \ref{npnotbeneficial}, a college can only beneficially manipulate college-proposing DAA by not proposing to student(s) it is matched with under $\mu$. Since all students $s \in S$ are acceptable to all colleges $c \in C$, a college proposes to students in order of its preferences until either its capacity is filled or it proposed to all students. Therefore any college that does not fill its capacity \emph{must} propose to all students.
\end{proof}
 For colleges $c$ with $|\mu(c)|<q(c)$ the question of whether this college can beneficially manipulate DAA can therefore trivially be answered without any further computations. The same applies to all colleges $c$ with capacity $q(c)=1$, as implied by the following lemma. 
%
%
In sum, only colleges $c$ with $|\mu(c)|=q(c)>1$ can potentially benefit from manipulating college-proposing DAA. A beneficial manipulation is achieved by misrepresenting preferences s.t. $c$ does not propose to a set of students $R \subseteq \mu(c)$ where $|R| \geq 1$. 
It is easy to see that if it is at all possible for $c$ to misrepresent its preferences s.t. it does not propose to any student $r \in R$, then this can also be achieved by $c$ misrepresenting its preferences s.t. all students $s \in S \setminus R$ are listed according to the truthful preferences, followed by all students $r \in R$ in arbitrary order. 

We now use these insights to show that the problem of deciding whether a given college $c \in C$ can beneficially manipulate college-proposing DAA via misrepresenting its preferences when all other agents state their preferences truthfully can be solved via at most $2^{q(c)-1}-1$ executions of DAA and is thereby fixed-parameter tractable w.r.t. the parameter $q(c)$.
\begin{theorem}
\label{thm:findingmanip}
Let $S, C$ be the agents of a many-to-one matching market. At most $2^{q(c)-1}-1$ executions of DAA are necessary in order to determine whether a given college $c \in C$ can beneficially manipulate the outcome of DAA via misrepresenting preferences.
\end{theorem}
\begin{proof}
Let $S, C$ be the agents of a many-to-one matching market and $\mu$ the matching that is obtained by applying DAA to this market. According to Lemma \ref{npnotbeneficial}, a college can only beneficially manipulate DAA by not proposing to student(s) it is matched with under $\mu$. Further, by Lemma \ref{onlyfullcapacity}, only colleges $c \in C$ with $|\mu(c)|=q(c)>1$ are able to achieve a beneficial manipulation.

Let $c \in C$ s.t. $|\mu(c)|=q(c)>1$, $s_l$ be the least preferred, according to $>_c$, student among $\mu(c)$ and $R \subseteq \mu(c)$ s.t. $s_l \in R$. Assume that $c$ can beneficially manipulate DAA by misrepresenting its preferences s.t. all students $s \in S \setminus R$ are listed according to the truthful preferences, followed by all students $r \in R$ in arbitrary order. According to Lemma \ref{npnotbeneficial}, $c$ does not propose to any student $u$ s.t. $s_l >_c u$ during the manipulated execution of DAA. Therefore the same beneficial manipulation can be achieved by $c$ misrepresenting its preferences s.t. all students $s \in \{s_l\} \cup S \setminus R$ are listed according to the truthful preferences, followed by all students $r \in R \setminus \{s_l\}$ in arbitrary order.
A beneficial manipulation for $c$, if one exists, can therefore be found by picking non-empty subsets $R \subseteq \mu(c) \setminus \{s_l\}$ and computing the outcome of DAA if $c$ states its preferences s.t. all students $s \in S \setminus R$ are listed according to the truthful preferences, followed by students $r \in R$ in arbitrary order. Since $|\mu(c) \setminus \{s_l\}|=q(c)-1$, DAA must be executed for at most $2^{q(c)-1}-1$ sets $R$.
\end{proof}

 
Since every possible manipulation of a college~$c$ under \collegeproposing DAA corresponds to a non-empty subset~$R\subset \mu(c)$, we can find an \emph{optimal} manipulation by checking every possible subset. Just as in the proof of \thmref{thm:findingmanip}, this implies a bound of $2^{q(c)-1}-1$ on the number of DAA executions. 

\begin{corollary}
	\label{cor:colprop-opt-fpt}
	Finding a college's~$c$ optimal manipulation by misrepresenting preferences under \studentproposing DAA is fixed-parameter tractable w.r.t. the parameter~$q(c)$.
\end{corollary} 




\section{Experimental Results}
\label{section:simulations}
%
%

We performed a series of simulations on randomly generated markets to estimate the fraction of instances of matching markets that can be beneficially manipulated by a single college if all other agents state their preferences truthfully.
%
DAA was first executed under the truthful preferences once, in order to obtain the resulting matching $\mu$. For each college $c \in C$ s.t. $|\mu(c)|=q(c)>1$ we then computed all non-empty subsets $R \subseteq \mu(c)\setminus\{s_l\}$, where $s_l$ is the least preferred, according to $>_c$, student among $\mu(c)$.  Subsequently, DAA was executed with $c$ misrepresenting its preferences s.t. all students $s \in S \setminus R$ are listed according to the truthful preferences, followed by all students $r \in R$ for different sets $R$, until either a beneficial manipulation by $c$ was found or all sets $R$ were checked in this manner.
In our simulations, the fraction of manipulable instances of matching markets lay between 5.7\% and 81.2\% for \studentproposing (36.53\% on average) and between 13.6\% and 100\% for \collegeproposing DAA (69.93\% on average). These experimental results suggest that the fraction of matching markets that can be beneficially manipulated is significant, independently of whether the student-proposing or the college-proposing version of DAA is used.

Preferences of both students and colleges were generated using the PrefLib tool by \citeauthor{MaWa13a}. The tool enables the generation of preference data using different methods, among which we chose those of \emph{impartial cultures} (IC) and \emph{Mallows mixtures} (MM) \citep{Mall57a}. Under IC, each (complete) preference list has the same probability, i.e. preferences are generated \emph{uniformly at random}. This is a standard method for generating random preference profiles for simulations, which was for example also used by \citet{TST01a} as well as by \citet{RoPe99a}. 
Preference profile generation using MM is parametrized by a set of \emph{reference rankings} $\mathcal{R}$ and a \emph{noise parameter} $\phi \in [0,1]$. 
We consider markets with $100$ to $200$ students and $15$ to $30$ colleges.
These numbers were chosen to allow us to compare the degree of manipulability of markets with different ratios of $|S|$ to $|C|$, while keeping the computational effort involved within reasonable limits. 
For each of these markets, we generated $1000$ sets of preference profiles each using IC as well as MM with a number of different parameter values. The capacities of colleges were generated using the following two methods. 
\begin{enumerate}[noitemsep,topsep=0pt,parsep=0pt,partopsep=0pt]
\item Capacities were generated uniformly at random between $1$ and $\left \lceil \frac{|S|}{|C|} \right \rceil$.
\item Capacities were initialized using method 1. While $\sum_{c \in C} q(c) < |S|$, a college was selected uniformly at random and its capacity increased by $1$.
\end{enumerate}
We applied both methods to each generated preference profile. This allows us to compare the manipulability of markets in which it is likely that some students remain unmatched (method 1) to those in which each student is guaranteed to be admitted to some college (method 2).
To allow for a better comparison of the degree of manipulability by colleges in \studentproposing and \collegeproposing DAA, the same preference profiles were used to simulate DAA in both kinds of markets.
In all simulations, the percentage of instances that can be beneficially manipulated is considerably higher under college-proposing DAA than under student-proposing DAA. On average 33.42\%, at least 7\% and up to 70.7\% more instances are manipulable under college-proposing DAA than under student-proposing DAA. Instances that are manipulable by at least one college, however, can on average be beneficially manipulated by less colleges under college-proposing DAA (9.07\% on average) than under student-proposing DAA (16.61\% on average). 

\section{Summary and Future Work}
\label{section:summary}


In this paper, we examine two issues regarding colleges manipulating DAA: frequency of manipulable instances and the complexity of manipulation. 
We investigated both issues w.r.t. manipulation by colleges in many-to-one markets, in which participants are required to submit \emph{complete preference lists}. 
We summarize our contributions in the context of previously existing results in Table~\ref{complexitysummary}.
\begin{table}[htb]
			\scalebox{1}{
\begin{tabular}{p{2.8cm} l l}
	\toprule
DAA	variant			& proposer		& proposee 				\\[.1cm] \midrule
one-to-one		& strategyproof	& in P (\citet{TST01a})	\\[.1cm]
many-to-one 	& \multirow{2}{*}{strategyproof}	& \multirow{2}{*}{in \textbf{P} (Thm. \ref{thm:o2mInP})}	\\
(student-proposing) & & \\[.1cm]
many-to-one 	& \multirow{2}{*}{\textbf{FPT} (Thm. \ref{thm:findingmanip})}	& \multirow{2}{*}{in P (from one-to-one)}	\\
(college-proposing) & & \\[.1cm]
many-to-many	& \textbf{FPT} (from Thm. \ref{thm:findingmanip})	& in \textbf{P} (from Thm. \ref{thm:o2mInP})\\
\bottomrule
\end{tabular}
}
\caption{Computational complexity of manipulating DAA by \emph{misrepresenting preferences} in markets with complete preference lists, strict and responsive preferences and fixed capacities. The results in bold are from this paper.}
\label{complexitysummary}
\end{table}
 Based on our experimental results on markets, we concluded that a significant fraction of instances of matching markets can be beneficially manipulated by at least one college, \emph{independently} of whether student-proposing or college-proposing DAA is used. 
 Our results indicate that \collegeproposing DAA is significantly \emph{more likely} to be beneficially manipulable. However, the average percentage of colleges that could benefit from a manipulation, given that the market is manipulable by at least one college, was \emph{considerably lower} under \collegeproposing DAA.
Another interesting result concerns the capacities of colleges.
In all our simulations, markets in which the capacities of colleges were chosen such that every student is guaranteed to be admitted to some college were considerably more prone to manipulation than those in which some students remained unmatched. 
A  natural extension of our work is to investigate whether these results transfer to matching markets in which agents are allowed to submit preference lists of \emph{arbitrary length}.
So far, no complexity results regarding manipulation by colleges in these kinds of markets exist. 


		\subsection*{Acknowledgments}
		NICTA is funded by the Australian Government through the Department of Communications and the Australian Research Council through the ICT Centre of Excellence Program.\\
		


\end{document}